\pdfoutput=1  
\documentclass[]{article}  
\usepackage{url,float}
\usepackage{graphicx}
\usepackage{amsmath}
\usepackage{amsfonts}
\usepackage{amssymb}
\usepackage{latexsym}

\newcommand{\hide}[1]{}

\newcommand{\ABox}{
\raisebox{3pt}{\framebox[6pt]{\rule{6pt}{0pt}}}
}
\newenvironment{proof}{{\bf Proof:}}{\hfill\ABox}

\newtheorem{theorem}{{\bf Theorem}}

\newtheorem{lemma}{Lemma}

\newcommand{\lemlab}[1]{\label{lemma:#1}}
\newcommand{\thmlab}[1]{\label{thm:#1}}
\newcommand{\figlab}[1]{\label{fig:#1}}
\newcommand{\seclab}[1]{\label{sec:#1}}

\newcommand{\lemref}[1]{\ref{lemma:#1}}
\newcommand{\thmref}[1]{\ref{thm:#1}}
\newcommand{\secref}[1]{\ref{sec:#1}}
\newcommand{\figref}[1]{\ref{fig:#1}}

\def\P{{\mathcal P}}
\def\G{{\Gamma}}

\def\d{{\delta}}
\def\a{{\alpha}}
\def\b{{\beta}}
\def\L{{\Lambda}}

{\makeatletter
 \gdef\xxxmark{%
   \expandafter\ifx\csname @mpargs\endcsname\relax 
     \expandafter\ifx\csname @captype\endcsname\relax 
       \marginpar{xxx}
     \else
       xxx 
     \fi
   \else
     xxx 
   \fi}
 \gdef\xxx{\@ifnextchar[\xxx@lab\xxx@nolab}
 \long\gdef\xxx@lab[#1]#2{{\bf [\xxxmark #2 ---{\sc #1}]}}
 \long\gdef\xxx@nolab#1{{\bf [\xxxmark #1]}}
 \gdef\turnoffxxx{\long\gdef\xxx@lab[##1]##2{}\long\gdef\xxx@nolab##1{}}%
}

\begin{document}

\title{%
Source Unfoldings of Convex Polyhedra\\
via Certain Closed Curves\footnote{
A preliminary abstract appeared in~\protect\cite{iov-sucpr-09}.
}
}

%
%
%
%
%

\author{%
Jin-ichi Itoh%
    \thanks{Dept. Math.,
        Faculty Educ., Kumamoto Univ.,
        Kumamoto 860-8555, Japan.
    \protect\url{j-itoh@kumamoto-u.ac.jp}}
\and
Joseph O'Rourke%
    \thanks{Dept. Comput. Sci., Smith College, Northampton, MA
      01063, USA.
      \protect\url{orourke@cs.smith.edu}.}
\and
Costin V\^{i}lcu%
    \thanks{Inst. Math. `Simion Stoilow' Romanian Acad.,
P.O. Box 1-764,
RO-014700 Bucharest, Romania.
    \protect\url{Costin.Vilcu@imar.ro}.}
}

\maketitle


\begin{abstract}
We extend the notion of a source unfolding of a convex polyhedron $\P$ 
to be based on
a  closed polygonal curve $Q$ in a particular class rather than based on a point.
The class requires that $Q$ ``lives on a cone'' to both sides; 
it includes simple, closed quasigeodesics.
Cutting a particular subset
of the cut locus of 
$Q$ (in $\P$)
leads to a non-overlapping unfolding of the polyhedron.
This gives a new general method to unfold the surface of any convex polyhedron
to a simple, planar polygon.
\end{abstract}

\section{Introduction}
\seclab{Introduction}
Two general methods were known to unfold the surface $\P$                      
of any convex polyhedron to a simple polygon in the plane:                      
the source unfolding and the star unfolding,                
both with respect to a point $x \in \P$; see, e.g.,~\cite{do-gfalop-07}.
In~\cite{iov-sucpql-10} we defined a third general method,
the star unfolding
with respect to any member of a class of
simple, closed curves $Q$ on $\P$,
a class we called
``quasigeodesic loops.''  
Here we extend the source unfolding to 
be based on
a different class of simple, closed, polygonal
curves $Q$.
The intersection of the two classes includes simple, closed
quasigeodesics.
Thus both the source and star unfolding are now generalized from points
to these quasigeodesics, in all cases unfolding
$\P$ to a simple, planar polygon.

{\bf Cut Locus.}
The \emph{point source unfolding} cuts the \emph{cut locus} $C_{\P}(x)$ of
the point $x$:
the closure of the set of all those points $y$ to which there is
more than one shortest path on $\P$ from $x$.
The point source unfolding has been studied
for polyhedral surfaces since~\cite{ss-spps-86} 
(where the cut locus is called the ``ridge tree'').
Our method also relies on the cut locus, but now
the cut locus $C_{\P}(Q)$ with respect to $Q$.
The definition of $C_{\P}(Q)$ is the same: it is the closure of all
points to which there is more than one shortest path from $Q$.
Here it is analogous to the ``medial axis'' of a shape; 
indeed, the medial axis of a polygon
is the cut locus of the polygon's boundary.
As with the point source unfolding, our unfolding essentially cuts
all of $C_{\P}(Q)$, but this statement needs to be qualified: we do not cut
some segments of the cut locus (incident to $Q$), and we cut some
additional segments not in the cut locus (again incident to $Q$).

{\bf Convex Curves and Quasigeodesics.}
Let $p$ be a point on an 
oriented, simple, closed, polygonal curve $Q$ on $\P$.
Let $L(p)$ be the total surface angle incident to the left side of $p$,
and $R(p)$ the angle to the right side.
$Q$ is a \emph{convex curve} if $L(p) \le \pi$ for all
points of $Q$.
$Q$ is a \emph{quasigeodesic} if $L(p) \le \pi$ and $R(p) \le \pi$
for all $p$,
i.e., it is convex to both sides.
Quasigeodesics, introduced by Alexandrov (e.g., \cite[p.16]{az-igs-67}) 
are the natural generalization of geodesics to polyhedral surfaces.
A \emph{quasigeodesic loop} has a single exceptional point at which
the quasigeodesic angle condition to one side may not hold.
We have the inclusions:
\{convex\} $\supset$ \{quasigeodesic loop\} $\supset$ \{quasigeodesic\}.

The class of curves for which our source unfolding method works
includes (a)~convex curves that pass through at most one vertex of $\P$,
and (b)~quasigeodesics.  The method does not (always) work for all
quasigeodesic loops.  Thus the class of curves for which the source and
star unfolding methods work are not directly comparable, but they both
include quasigeodesics.

{\bf Curves ``Living on a Cone.''}
The precise class of curves for which our source unfolding method
works depends on the following notion.
Let $Q$ be (as before) an oriented, simple, closed, polygonal curve on $\P$.
Let $N$  be a vertex-free neighborhood of $Q$ in $\P$ to the left of and
bounded by $Q$.
We say $Q$ \emph{lives on a cone} to its left if there exists
a cone $\L$ and an $N$ such that $Q \cup N$ may be embedded isometrically
on $\L$, enclosing the cone apex $a$.
A \emph{cone} is a developable surface with curvature zero everywhere
except
at one point, its apex $a$.  We consider a planar polygon to be a cone with
apex angle $2\pi$, and a cylinder to be a cone with apex angle $0$.
The source unfolding described in this paper
works for any curve $Q$ that (a)~lives on a cone
to both sides (perhaps on different cones), and (b)~such that each point
of $Q$ is ``visible'' from the apex $a$ along a generator of the cone
(a line through $a$ lying in $\L$).
See Figure~\figref{Cone3D}.
\begin{figure}[htbp]
\centering
\includegraphics[width=0.5\linewidth]{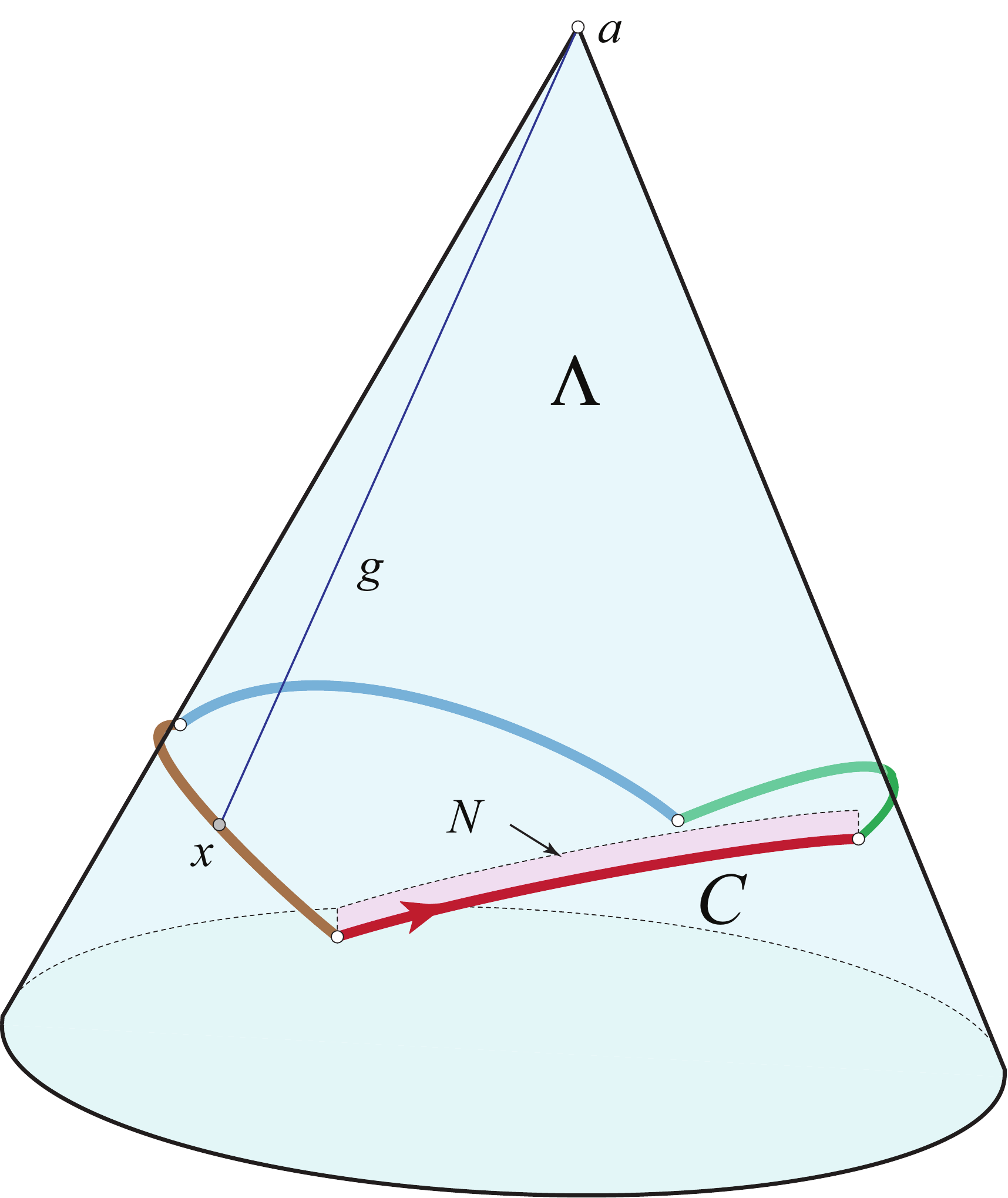}
\caption{A 4-segment curve $Q$ which lives
on cone $\L$ to its left.
A portion of $N$ is shown, 
and a generator $g=ax$ is illustrated.
(Adapted from~\protect\cite{ov-ceccc-11}.)
}
\figlab{Cone3D}
\end{figure}
We should remark that the cone on which a curve $Q$ lives has no direct
relationship (except in special cases) to the surface that results
from extending the faces of $\P$ crossed by $Q$.
The cones on which $Q$ lives play a central role in our proof technique.
Although a curve could live on many different cones, 
it is established in~\cite[Lem.~3]{ov-ceccc-11}
that the cone is uniquely determined to each side by $\P$.
We also established~\cite[Thm.~3]{ov-ceccc-11}
that the classes of curves listed above live on cones in this sense.
Indeed the set of curves enjoying the properties~(a) and~(b) above
is wider than what we list here, but as
the full class of curves that live on cones to both sides is not yet
precisely
delimited, we leave this issue aside.

\section{Preview of Algorithm}
\seclab{Preview}
Assume we are given a curve $Q$ satisfying our conditions:
it is convex to one side, and it lives on cones to either side.
We now describe the unfolding abstractly at a high level.
First we need some notation.

$Q$ divides $\P$ into two closed ``halves,''
$P_1$ and $P_2$.
We handle the two halves a bit differently.
Let $P_1$ be the half to the convex side (say, the left side),
and $P_2$ the half to the (possibly) reflex (i.e., nonconvex) right side.
(If $Q$ is a quasigeodesic, both halves are convex.)
Let $C_{P_i}$ be the portion of the cut locus $C_{\P}(Q)$ in each
half $P_i$.

Cut all edges of $C_{P_1}$ not incident to $Q$,
and cut one additional 
precisely selected segment from $C_{P_1}$ to $Q$.
To the (possibly) nonconvex side $P_2$, cut all of $C_{P_2}$, including those
edges incident to $Q$.
In addition, we cut shortest path 
segments from $C_{P_2}$ to any remaining reflex vertices of $Q$.
See Figure~\figref{ConvexQUnf}.
The result, we will show, is an unfolding of $\P$ to a simple, planar polygon.

\begin{figure}[htbp]
\centering
\includegraphics[width=0.85\linewidth]{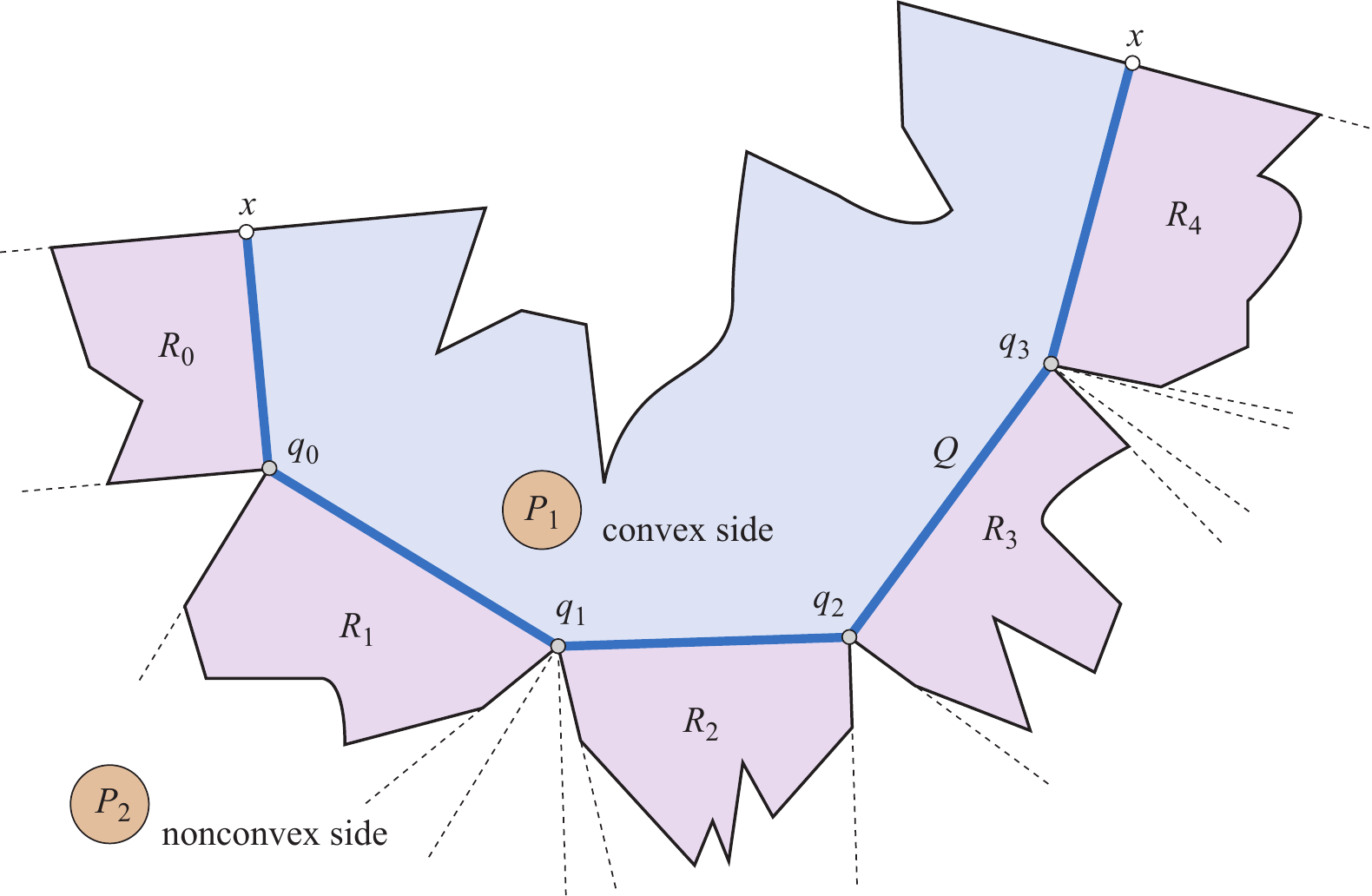}
\caption{An abstract depiction of the full source unfolding
for a convex curve $Q$. Here $q_1$ and $q_3$ are reflex vertices of $Q$ to the
$P_2$ side.
(The figure is not metrically accurate.)
}
\figlab{ConvexQUnf}
\end{figure}

We do not concentrate in this paper on algorithmic complexity issues,
which will only be touched upon in Section~\secref{Future}.
But the described procedure is a definite, finite algorithm which
works for any $Q$ in the appropriate class.

The unfolding procedure is best viewed as unfolding each half
separately, and then gluing them together along $Q$, as the examples
below will emphasize.

\begin{figure}[htbp]
\centering
\includegraphics[width=\linewidth]{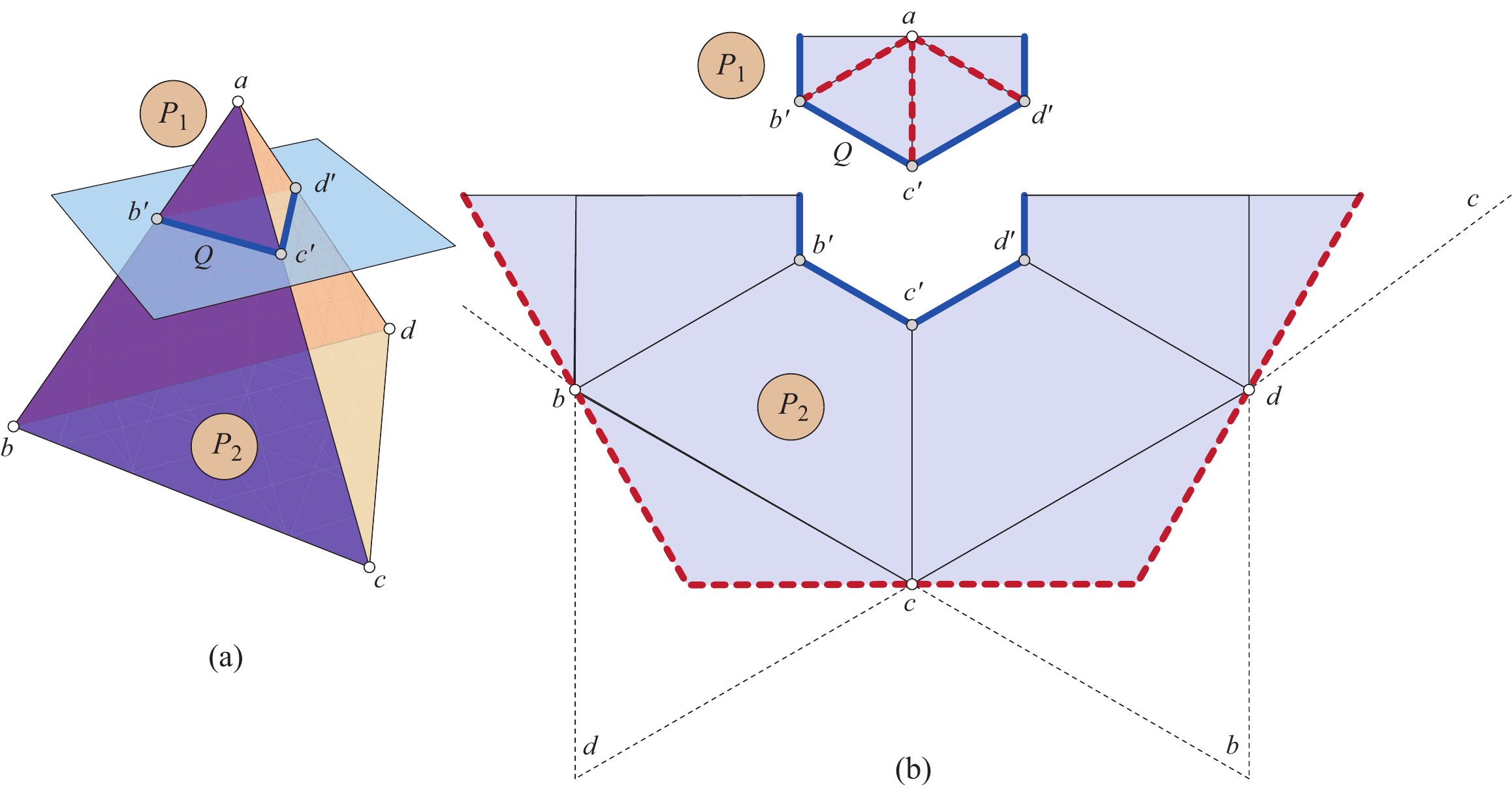}
\caption{(a)~Regular tetrahedron sliced at $\frac{2}{3}$'s of its altitude.
(b)~Unfoldings of $P_i$, with cut loci $C_{P_i}$ drawn dashed.
$C_{P_2}$ partitions the bottom face $\triangle bcd$ into three congruent triangles
meeting at its centroid.}
\figlab{TetraSlice}
\end{figure}


{\bf Examples.}
Throughout this paper we will use three example polyhedra $\P$
to illustrate concepts.
Here we introduce two; the third will be described in
Section~\secref{CutLocus}.

\emph{Example~1.}
The first example is a regular tetrahedron with a convex curve parallel
to the base;
see Figure~\figref{TetraSlice}(a),
where $Q=(b',c',d')$.
$Q$ lives on the same cone to each side, the cone determined by 
the lateral faces of the tetrahedron.
Note that the angle at each corner of $Q$ is convex ($\frac{2}{3} \pi$) to
one side and reflex ($\frac{4}{3} \pi$) to the other side.
One of our main results,
Theorem~\thmref{HalfSourceUnfolding},
shows that each half unfolds separately without overlap, as illustrated in Figure~\figref{TetraSlice}(b).
A second main result, Theorem~\thmref{SourceUnfolding},
shows that the two halves may be joined to one non-overlapping piece, in this case
producing a trapezoid.
This is an atypical case in many respects, but will be useful for that reason
to illustrate degeneracies.
For example, notice that no part of $C_{P_1}$ is cut in this example, because all
segments of that cut locus are incident to $Q$.

\emph{Example~2.}
Our second example is more generic:
a cube twice truncated, with
$Q$ the particular quasigeodesic shown in
Figure~\figref{CubeTrunc}(a).
The angles at
the vertices of $Q=(v_0,v_1,v_7,v_{10})$ within $P_1$
are, respectively, 
$(
\frac{3}{4} \pi,
\frac{1}{2} \pi,
\frac{3}{4} \pi,
\frac{1}{2} \pi
)$,
and the angles at those
vertices within $P_2$
are 
$(
\frac{3}{4} \pi,
\pi,
\frac{3}{4} \pi,
\pi
)$.
Because all of these angles are at most $\pi$, 
$Q$ is convex to both sides and so a quasigeodesic.
The cone on which $Q$ lives to the $P_2$-side is evident:
its apex is the cube corner truncated.
The cone on which $Q$ lives to the $P_1$-side is not evident;
it will be described later (in Figure~\figref{CubeCones}).
For $Q$ a simple closed quasigeodesic, the cut loci $C_{P_i}$ are each a single tree
with each edge a (geodesic) segment,
as illustrated in Figure~\figref{CubeTrunc}(b,c).
Theorem~\thmref{SourceUnfolding} leads
to the unfolding shown in Figure~\figref{CubeTrunc}(d).

Our third example will illustrate that the
edges of the cut locus can also be parabolic arcs.

\begin{figure}[htbp]
\centering
\includegraphics[height=0.95\textheight]{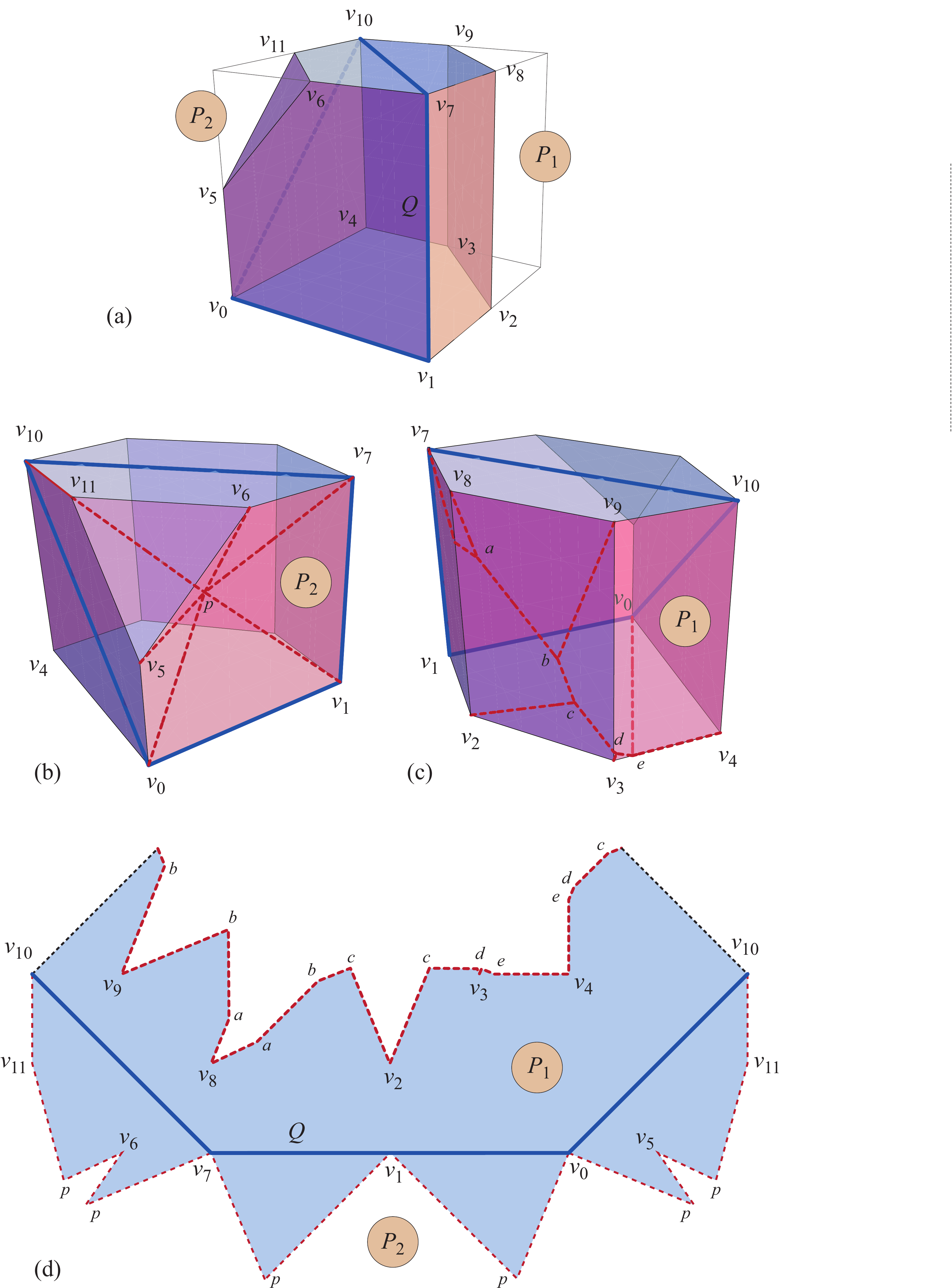}
\caption{
(a)~Truncated cube and quasigeodesic $Q=(v_0,v_1,v_7,v_{10})$.
(b,c)~Views of $P_i$ and cut loci $C_{P_i}$ (dashed).
(d)~Unfolding to a simple polygon.
}
\figlab{CubeTrunc}
\end{figure}

\section{Cut Locus of $Q$}
\seclab{CutLocus}
The proof that the described source unfolding avoids overlap
relies on two key ingredients:
the structure of the cut loci in each half of $\P$, and the cones
on which $Q$ lives.  Here we focus on the cut loci.

A \emph{vertex of $Q$} is a point $q$ with 
an angle at $q$ in $P_i$ different from $\pi$;
note this definition depends on the half $P_i$.
Thus if $Q$ is a geodesic, it has no vertices at all to either side.
It will be useful to distinguish two varieties
of these vertices:
a \emph{convex vertex} has angle less than $\pi$ in $P_i$,
and a \emph{reflex vertex} has angle greater than $\pi$.
Note the meaning of these is interchanged when looking
from $P_1$ compared to looking from $P_2$
(although it is possible that an angle differs from $\pi$ from 
one side and is equal to $\pi$ from the other,
e.g., $v_{10}$ in Figure~\figref{CubeTrunc}).
Let $q_0,q_1,\ldots,q_k$
be the vertices of $Q$
in some circular order, with respect to the half $P_i$.
(Although these vertices depend on $P_i$, we opt not to subscript with
$i$ to ease notation.)

The cut locus $C_{P_i}$ is a tree
whose leaves span the vertices of $P_i$, 
including the convex vertices of $Q$, 
which must be leaves of $C_{P_i}$.
(That $C_{P_i}$ is a tree is well-known in Riemannian geometry,
e.g., see~\cite[p.~539]{t-spts-98}. 
This can also be seen from the fact that each half $P_i$ has a finite intrinsic
diameter,
and so shortest paths from $Q$ are finite in length and each ends
at a cut point.)
Each non-leaf point $p$ of $C_{P_i}$ has at least two shortest
  paths
to $Q$; leaves have precisely one shortest path to $Q$,
or might lie on $Q$.
Each shortest path from $p \in C_{P_i}$ to $Q$ is called a 
\emph{projection} to $Q$.
The segment of $C_{P_i}$ incident to each convex vertex of $Q$ bisects
the angle there.

We now review our three examples to illustrate these relationships.
In Figure~\figref{TetraSlice}, $C_{P_1}$
is a tree spanning the apex $a$ and the three vertices
$\{b',c',d'\}$ of $Q$.
$C_{P_2}$ is a tree in the bottom face spanning its
three vertices $\{b,c,d\}$ (but not touching $Q$).

Figures~\figref{CubeTrunc}(b,c) show that $C_{P_i}$ are
each trees spanning the vertices of $P_i$ and touching $Q$:
at its four vertices in $P_2$, and its two vertices in $P_1$
(the angles at $v_1$ and at $v_{10}$ are $\pi$ in $P_1$).

\emph{Example~3.}
Figure~\figref{SquarePyramid2} shows an example that illustrates
several aspects not present in the other two examples.
First, $Q$ is neither a convex curve nor a quasigeodesic,
but it nevertheless lives on a cone to both sides, namely, the cone
determined by the pyramid's lateral sides.  Thus our method applies to
this $Q$.
Second, $C_{P_1}$ is a tree spanning the apex $a$ and the
two convex vertices of $Q$ to that side, $\{b',d'\}$.
$C_{P_1}$ includes four parabolic arcs:
one generated by the edge $b'e'$ and the reflex vertex $c'_1$,
one generated by the edge $c'_1 c'$ and the reflex vertex $e'$,
and two more symmetrically placed arcs.
$C_{P_2}$ is a tree that includes an {\tt X} on the bottom face
spanning $\{b,c,d,e\}$.
In general, parabolic arcs arise as 
arcs at equal distance to
a reflex
vertex and an edge.

\begin{figure}[htbp]
\centering
\includegraphics[width=\linewidth]{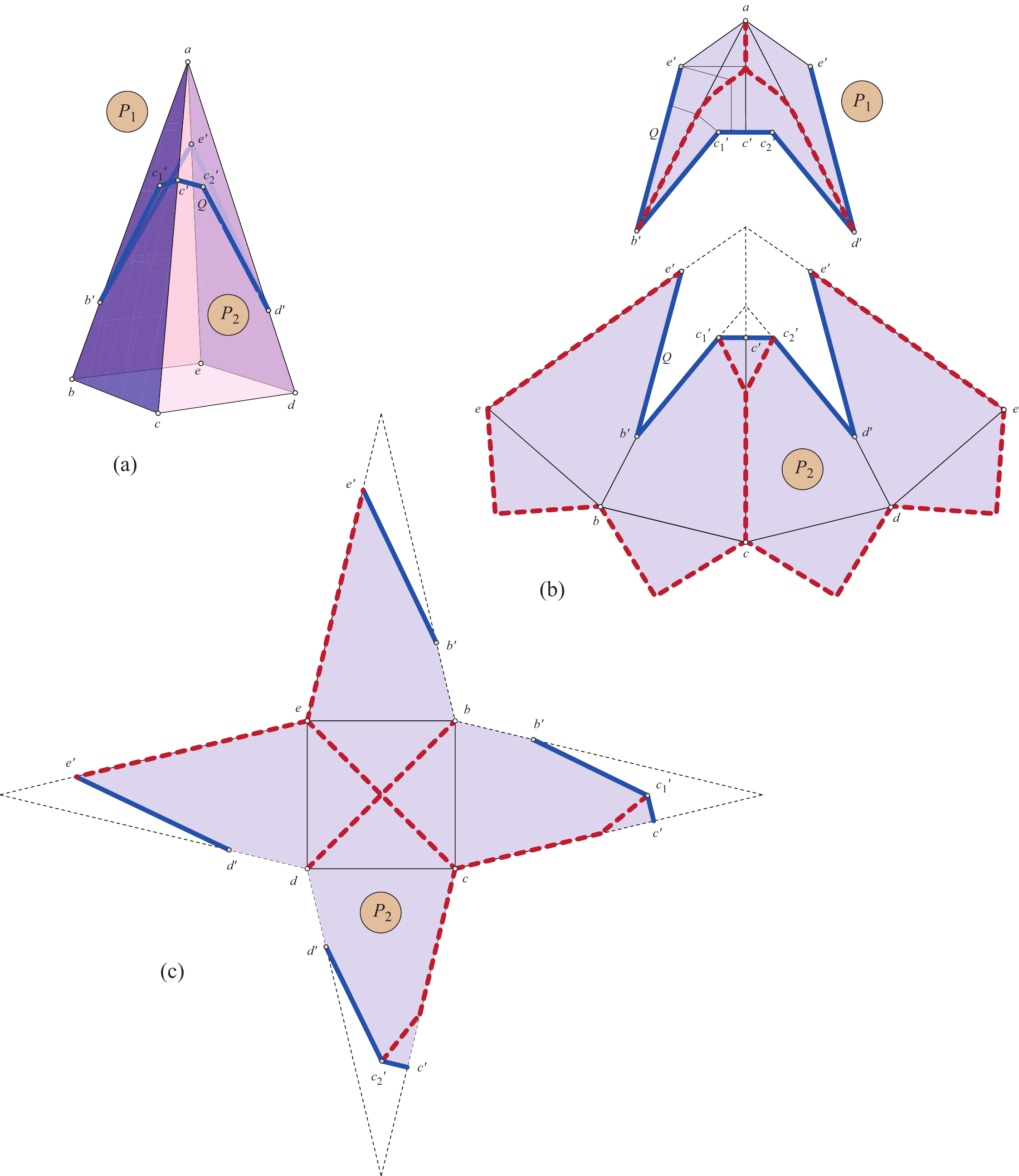}
\caption{(a)~A square-based right pyramid, with altitude twice
the base side length.
$|b b'| = |d d'| = |a e'| = \frac{1}{4}|a b|$,
and $|b' c'_1|= \frac{3}{4}|b'e'|$.
(b)~Source unfolding of the two halves $P_1$ and $P_2$.
(c)~A view of $P_2$ from the bottom of the pyramid.
}
\figlab{SquarePyramid2}
\end{figure}

\section{Cut Loci on Cones}
\seclab{Cones}
The essence of our proof, that the unfolding of each half $P_i$ avoids
overlap 
(Theorem~\thmref{HalfSourceUnfolding}),
is that pieces of $P_i$ embed into the cone $\L_i$,
and that $\L_i$ develops without overlap in the plane when cut along a generator.
A key to our approach is to define these ``pieces'' of $P_i$ (in Sec.~\secref{Peels}) by 
comparing the cut locus $C_{P_i}$ on $P_i$ with the
cut locus $C_{\L_i}$ on the cone $\L_i$ on which $Q$ lives to the $P_i$-side.
The cone $\L_i$ and surface $P_i$
share the same boundary $Q$,
and by construction, the same angles occur along $Q$.
Therefore, $Q$ has the same vertices in both $\L_i$ and $P_i$.
Thus $C_{P_i}$ and $C_{\L_i}$ both have the same set of leaves
touching $Q$, but of course in general they differ in the interior
of the surfaces.

To one of the two sides, the cone may be unbounded 
(established in~\cite{ov-ceccc-11}).
This occurs in both 
Figure~\figref{TetraSlice} and
Figure~\figref{SquarePyramid2}.
For the tetrahedron (Figure~\figref{TetraSlice}), $\L_2$ is unbounded,
and $C_{\L_2}$ is empty.
For the pyramid (Figure~\figref{SquarePyramid2}),  $\L_2$ is again unbounded,
and $C_{\L_2}$  is in this case a forest, 
including two halfline branches to infinity,
one from $e'$ through $e$, and another from the {\tt Y}-junction
below $c'$ through $c$.
The cut locus can only be a forest (as opposed to a tree) to
an unbounded side.  This is because an unbounded surface $\L_i$ contains
infinite-length
shortest paths without cut points, which then separate $C_{\L_i}$ into
components, as in the pyramid example.
Only one side may be unbounded, unless the cone is a cylinder.

In the truncated cube example, both cones are 
bounded (because the quasigeodesic $Q$ is convex to both sides),
and both $C_{\L_1}$ and  $C_{\L_2}$ are trees; see
Figure~\figref{CubeCones}.
Note that the cut locus extends to the apex $a_i$ of $\L_i$
in both instances,
which it must because the cut locus includes all vertices.

\begin{figure}[htbp]
\centering
\includegraphics[width=\linewidth]{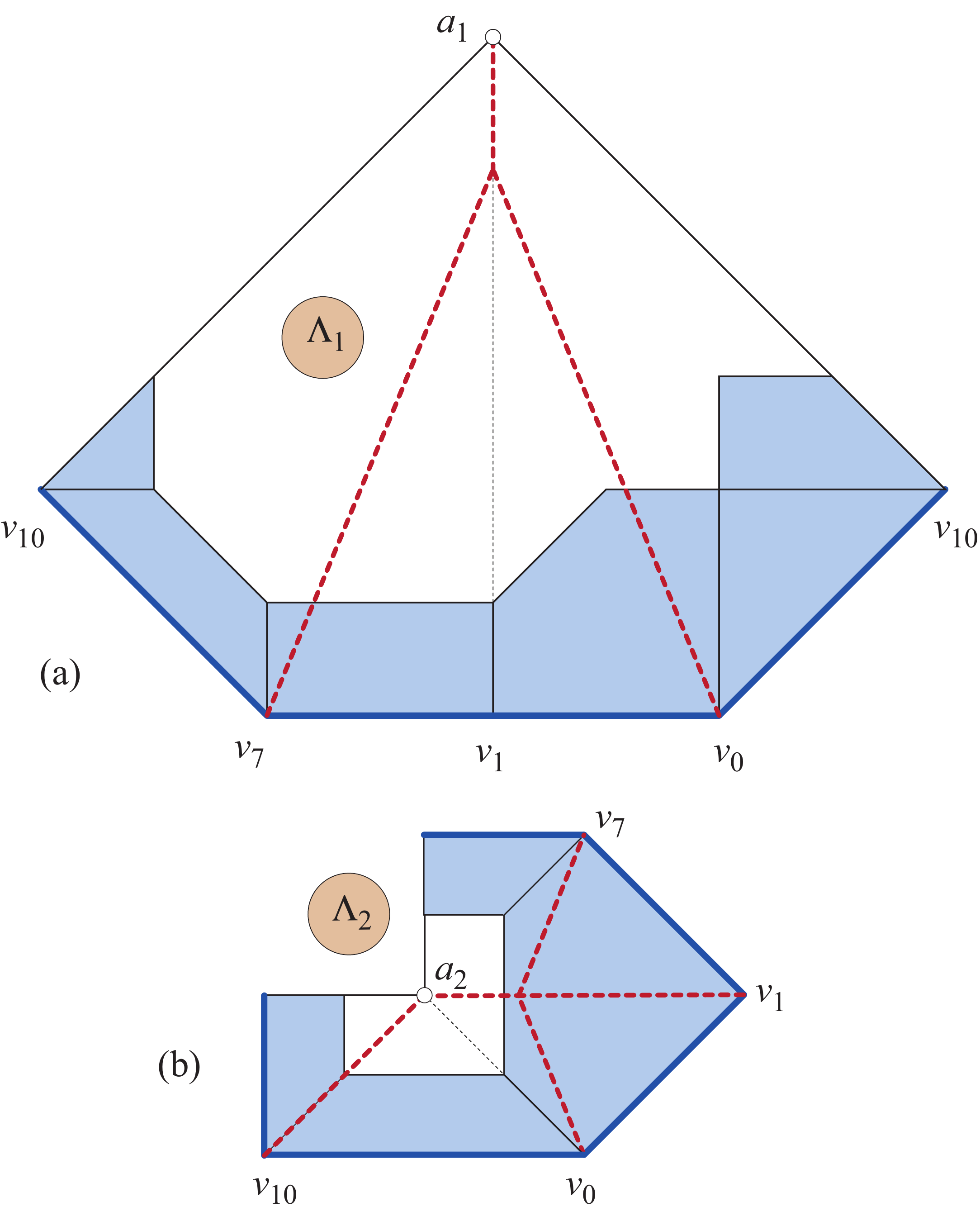}
\caption{
Developed cones $\L_1$ and $\L_2$ 
for $\P$ in Fig.~\protect\figref{CubeTrunc}(b,c), 
with $C_{\L_1}$ and $C_{\L_2}$ indicated.
}
\figlab{CubeCones}
\end{figure}

$C_{\L_i}$ is determined by any small neighborhood of $Q$ to the $P_i$-side,
because $\L_i$ is so determined (see again~\cite{ov-ceccc-11}).
And because every leaf of $C_{\L_i}$ is also a leaf of $C_{P_i}$,
the bisecting property of cut loci 
(each edge of the cut locus incident to a vertex $v \in Q$
bisects the angle of $P_i$ at $v$)
implies that
small neighborhoods of the leaves of $C_{\L_i}$
are included in $C_{P_i}$.  
In other words, the edges of $C_{P_i}$ and $C_{\L_i}$ issuing from
vertices of $Q$ coincide until they hit a vertex of $C_{P_i}$ or a
vertex of $P_i$.
We use this property 
in the proof of Lemma~\lemref{Nesting} below.

We now turn to defining the ``pieces'' of $P_i$ that embed in $\L_i$.

\begin{figure}[htbp]
\centering
\includegraphics[width=\linewidth]{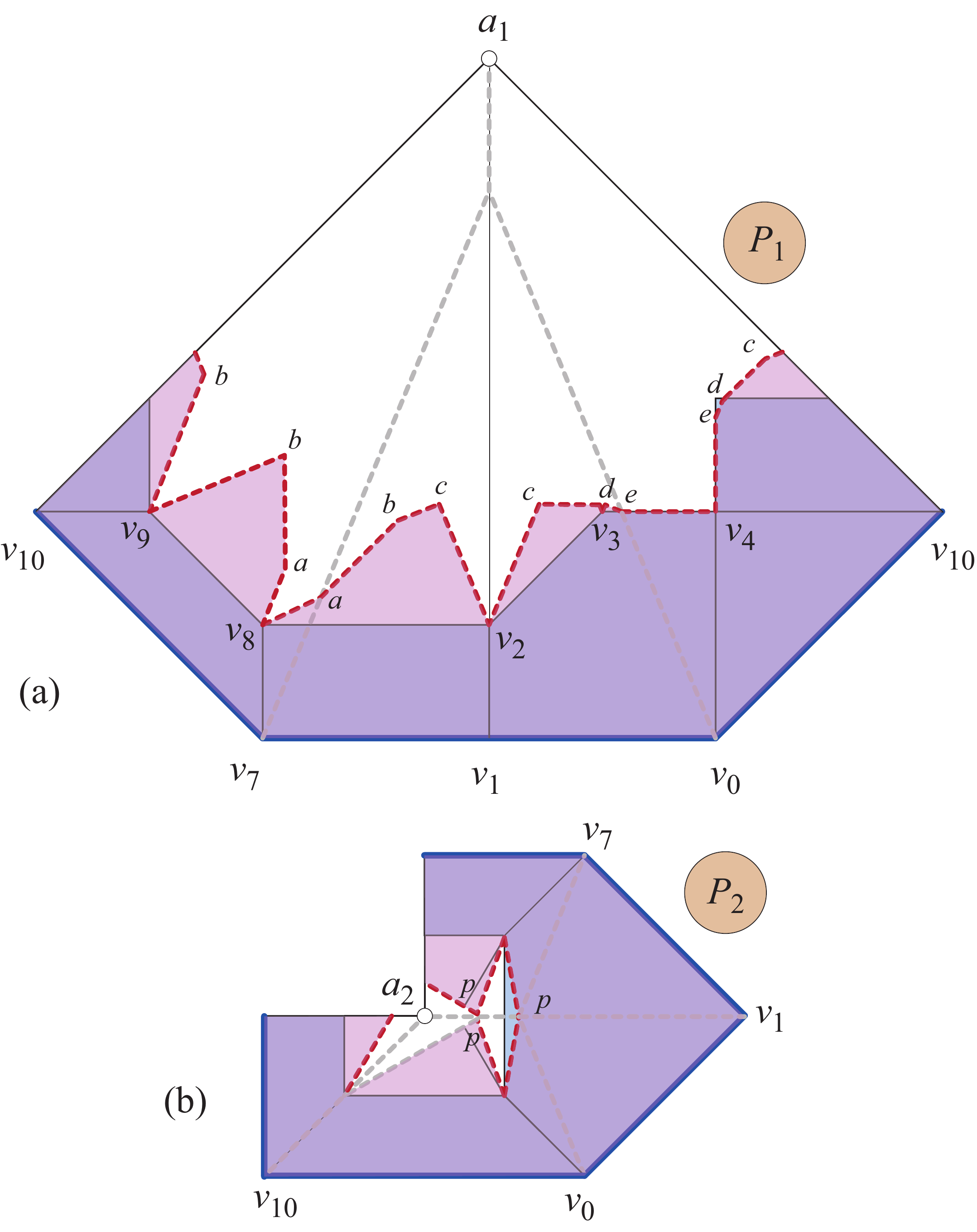}
\caption{
(a)~Nesting of cut $P_1$ within the unfolded $\L_1$ of Fig.~\protect\figref{CubeCones}(a).
(b)~Similar nesting of $P_2$ within $\L_2$ of Fig.~\protect\figref{CubeCones}(b).
}
\figlab{NestedPeels}
\end{figure}

\section{Peels \& Subpeels: Embedding in the Cone}
\seclab{Peels}
Let $u_0,u_1,\ldots,u_m$ be the vertices
(leaves and junction points) of $C_{P_i}$,
following a circular ordering of all of their their 
projections to $Q$.
Note that this ordering is unambiguous even though some points 
have multiple equal-length projections to $Q$,
because these projections never cross.
Therefore, those leaves and junction points of $C_{P_i}$ appear
several times in the 
sequence $u_0,u_1,\ldots,u_m$, each as many times as its number of projections.
Let $(u_j,u_k)$ be two consecutive leaves
of $C_{P_i}$, and $(u'_j,u'_k)$ corresponding consecutive projections
onto $Q$.
The \emph{peel} $\a_{P_i}(u_j,u_k)$ is the closed flat region of $P_i$
bounded by the two projection paths $u_j u'_j$, $u_k u'_k$,
the subpath $Q_{ij}$ of $Q$ from $u'_j$ to $u'_k$,
and the unique path in $C_{P_i}$ connecting $u_j$ to $u_k$,
such that $\a_{P_i}(u_j,u_k)$ contains no leaf of $C_{P_i}$.
Each peel $\a_{P_i}(u_j,u_k)$ is isometric to a planar convex polygon
if $Q_{ij}$ is convex to the left;
otherwise, $\a_{P_i}(u_j,u_k)$ can be decomposed into the union of planar convex polygons and triangles whose base is a parabolic arc and whose vertex opposite to that side is a reflex vertex of $Q$.

Between two consecutive leaves $u_j$ and $u_k$ 
are the vertices along the tree path, $u_j, u_{j+1},\ldots,u_{k-1},u_k$.
Each of these delimits a \emph{subpeel} of $\a_{P_i}(u_j,u_k)$, partitioning
the peel along the projection segments $u_l u'_l$, $j < l < k$.

The notion of peel and subpeel can be defined
analogously for $\L_i$.
If $C_{\L_i}$ has at most one component, everything defined above for
$C_{P_i}$ holds exactly as stated.
Assume now that $C_{\L_i}$ has at least two components.
In this case, all leaves of $C_{\L_i}$ belong to $Q$, and we consider them as projection points of themselves.
Again we consider all projections of the junction points of $C_{\L_i}$.
We take the circular order along $Q$ of all these projection points.
If two consecutive projections correspond to points on the same component of $C_{\L_i}$, we use the definition of peels and subpeels given above.

Now assume two consecutive projections $u_j$ and $u_k$ correspond to points on different components of $C_{\L_i}$.
Then the peel $\a_{\L_i} (u_j , u_k )$ is the closed flat region of
$\L_i$ bounded by the two projection paths $u_j u'_j$ , $u_k u'_k$,
the subpath $Q_{ij}$ of $Q$ from $u'_j$ to $u'_k$ , and the two paths
in $C_{\L_i}$ from $u_j$, respectively $u_k$, to infinity,
such that $\a_{P_i}(u_j,u_k)$ contains no leaf of $C_{P_i}$.
Here we note that each component of $C_{\L_i}$ has precisely one arc going to infinity, so the definition above is correct.
The definition for subpeels is analogous: either we get bounded polygons, or unbounded polygons determined by arcs to infinity in different components of $C_{\L_i}$.

These peels can be decomposed into the union of 
(a)~planar rectangular trapezoids, with one side either a line-segment, or a parabolic arc, or ``an arc at infinity''; and
(b)~triangles whose base is a parabolic arc (possibly at infinity) and whose vertex opposite to that side is a reflex vertex of $Q$.


We now prove the central technical result, Lemma~\lemref{Nesting},
which establishes the embedding of $P_i$ into $\L_i$.
This lemma shows that the peels of $P_i$ nest inside the peels of $\L_i$.
At one spot in the argument, we need a specialized lemma,
which we invoke (Lemma~\lemref{Junction}) before proving it.
As a consequence of the nesting lemma, embedding of
the half-surfaces into their cones follows, as 
summarized in Lemma~\lemref{Embed}.

\begin{lemma}[Peel Nesting]
Each subpeel $\b_{P_i}$ of $P_i$ is isometric to a region of a 
subpeel $\b_{\L_i}$ of $\L_i$.
The union of the subpeels in one peel $\a_{P_i}$ of $P_i$ is
non-overlapping in some $\a_{\L_i}$,
and thus each peel $\a_{P_i}$ is nested inside a peel $\a_{\L_i}$ of $\L_i$.
\lemlab{Nesting}
\end{lemma}
\begin{proof}
Let $\G_{P_i}$ be the directed curve that traces around
the maximal subtree of $C_{P_i}$ disjoint from $Q$,
i.e., around the cut locus minus the edges incident to $Q$.
$\G_{P_i}$ is an Eulerian tour of this subtree,
tracing
its edges twice, once from each side.
Thus each non-leaf point $x$ of $C_{P_i} \setminus Q$ has at least
two images in $\G_{P_i}$, and a junction point $x$ of $C_{P_i}$ has
deg$(x) \ge 3$ images in $\G_{P_i}$.
We define $\G_{\L_i}$ to be the ``image'' of $\G_{P_i}$ on $\L_i$:
an isometric tracing with the same angles, on $\L_i$.
We are going to track a variable point $x_t$ on $\G_{\L_i}$ inside $\L_i$,
and analyze how $x_t$ interacts with $C_{\L_i}$.
The crux of the proof analyzes what happens when $x_t$ might leave
its peel $\a_{\L_i}$.

We will illustrate the proof with the portion
of $\G_{\L_1}$ for $P_1$ of our truncated cube example shown in Figure~\figref{JunctionProof}.
\begin{figure}[htbp]
\centering
\includegraphics[width=0.75\linewidth]{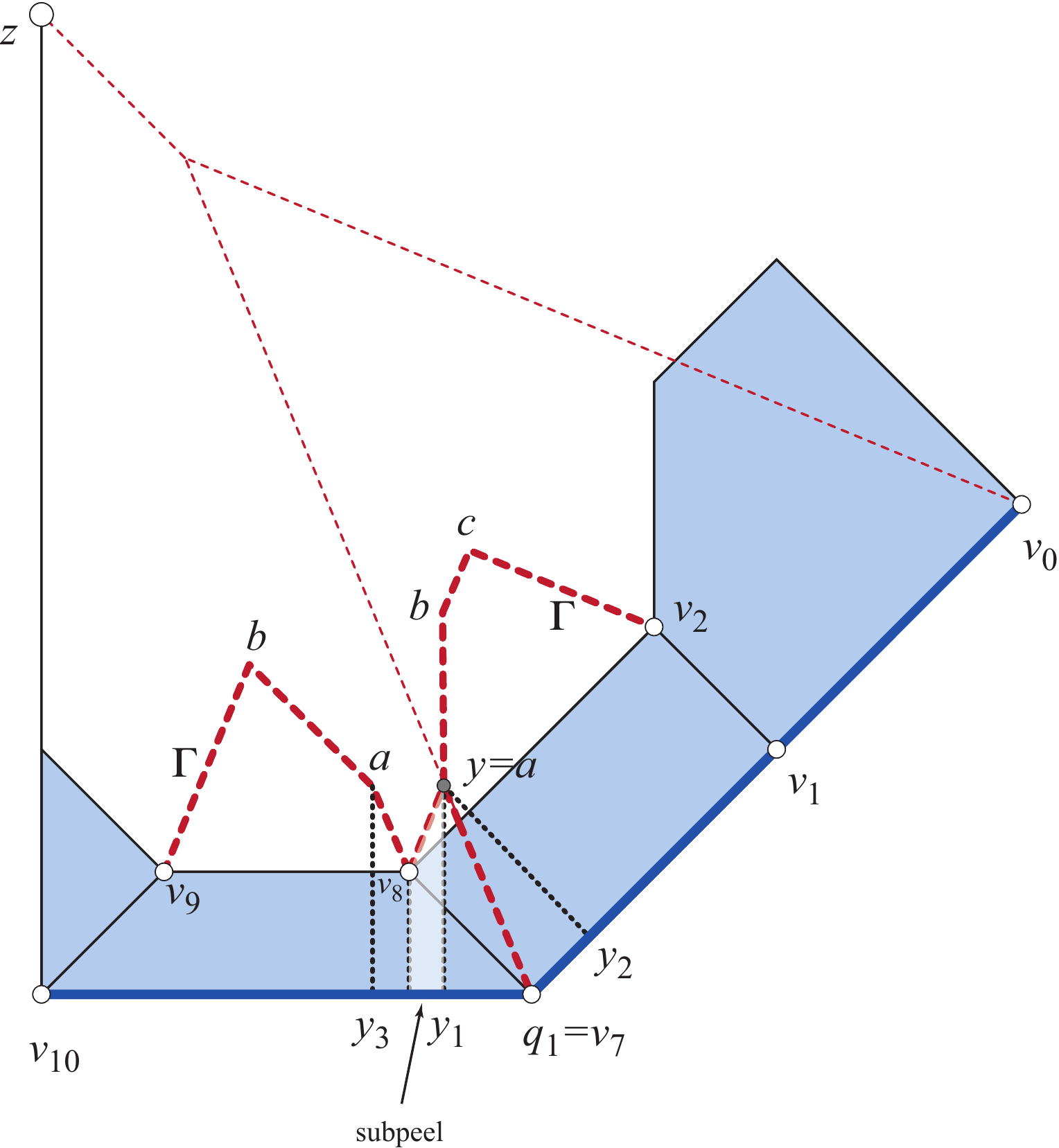}
\caption{$\G_{\L_1}$ is the dashed curve, directed left-to-right
(from $v_9$ to $v_8$).
The current subpeel is that bounded on the left by $v_8$.
$x_t$ crosses $C_{\L_1}$ at $y$, which coincides
with the junction point $a$ of $C_{P_1}$; see
Fig.~\protect\figref{CubeTrunc}(c).
The three segments are the projections of $y{=}a$ on $P_1$;
$y_1$ and $y_2$ are the only two projections of $y$ on $\L_1$.
$\G_{\L_1}$ enters the new subpeel
$\b_{\L_1}(v_7,y_2)$ and the new peel is $\a_{\L_1}(v_7,v_0)$.
}
\figlab{JunctionProof}
\end{figure}

To initiate the analysis,
let $x_0 \in \G_{\L_i}$ be any point in the interior of a peel $\a_{\L_i}$.
We need to argue for the existence of such a point.
We will choose a leaf of $C_{P_i}$, but it must be chosen with
some care.
First, if $C_{P_i}=C_{\L_i}$, then the claim of the lemma
is trivial.  So assume the cut loci differ.
Delete from $C_{P_i}$ all branches
in common with $C_{\L_i}$.
By ``branches'' here we mean subtrees of $C_{P_i}$ whose removal
does not disconnect $C_{P_i}$; so 
what remains is still a tree.


Let $v \in P$ be a leaf of this reduced tree.
As we argued above, the image in $\L_i$ of a small neighborhood of $v$
in $P_i$ remains included in $\L_i$.  
So we may take $x_0 = v$.
Vertex $v_8$ in Figure~\figref{JunctionProof} could serve as $x_0$.

Now we move a point $x_t$ along $\G_{\L_i}$ continuously from $x_0$.
As long as $x_t$ remains inside the peel $\a_{\L_i}$,
the subpeel $\b_{P_i}$ to which $x_t$ belongs remains included
in peel $\a_{\L_i}$.
Adjacent subpeels $\b_{P_i}$  and $\b'_{P_i}$
share a side orthogonal to $Q$ in $\a_{\L_i}$, and so they do not overlap
one another.

Now assume that $x_t$ reaches a point $x_t = y \in C_{\L_i}$.
If $\G_{\L_i}$ touches but does not cross $C_{\L_i}$,
then 
the subpeel $\b_{P_i}$ to which $x_t$ belongs,
or the adjacent subpeel $\b'_{P_i}$ into which $x_t$ is moving,
remains included in $\a_{\L_i}$,
and there is nothing to prove.
We should note that, in general, we cannot
conclude that touching-but-not-crossing $C_{\L_i}$ necessarily 
implies that $y$ is a junction of $C_{P_i}$.
It could be that $y$ is a leaf of $C_{\L_i}$, which is another
instance of touching-but-not-crossing, and again nesting remains clearly true.
So henceforth we assume $y$ is interior to  $C_{\L_i}$.

Now we consider the situation when $x_t$ crosses $C_{\L_i}$ at $y$,
as it does at $y{=}a$ in Figure~\figref{JunctionProof}.
This is exactly when the claim of the theorem might be false,
for the subpeel $\b_{P_i}$ bounded by $\G_{\L_i}$ then extends beyond the peel $\a_{\L_i}$.
But Lemma~\lemref{Junction} below establishes that $y$ must be 
(the image of) a junction
point of $C_{P_i}$
(and note that $a$ is a junction in our example).
This means that the subpeel $\b_{P_i}$ ends at $y$, and $\G_{\L_i}$ enters a new
subpeel $\b'_{P_i}$ of a new peel $\a'_{\L_i}$
of $\L_i$. 
Thus the subpeel nesting claim of
the lemma is established.
The peel nesting property will be obtained following the proof of Lemma~\lemref{Junction}.
{\hfill\ABox}\end{proof}

We complete the proof of Lemma~\lemref{Nesting} with a technical lemma
that was invoked above.
We use $\d_S(x,y)$ to represent the distance function on surface $S$:
the length of a shortest path on $S$ between $x$ and $y$.
We will employ the
Alexandrov-Toponogov Comparison Theorem
in two versions, Lemmas~\lemref{comp_1} and~\lemref{comp_2}
in the Appendix.

\begin{lemma}[Junction Crossing]
When $\G_{\L_i}$ crosses $C_{\L_i}$ at $x_t=y$,
$y$ is (the image of) a junction point of $C_{P_i}$.
\lemlab{Junction}
\end{lemma}
\begin{proof}
Because $\G_{\L_i}$ crosses $C_{\L_i}$ at $y$,
$y$ has at least two projections 
to $Q$ on $\L_i$
(because every interior point of $C_{\L_i}$ has at least two
projections).
One of the two segments, say $y y_1$, also is a segment on $P_i$,
because up to this point $x_t{=}y$ in our tracing,
the peel $\a_{P_i}$ containing $x_t$ is nested in $\a_{\L_i}$.
So the length of these segments is the
same on $P_i$ and on $\L_i$: $\d_{P_i}(y,y_1)=\d_{\L_i}(y,y_1)$.

Let $y_2$ be the ``next'' projection of $y$ to $Q$ on $\L_i$,
i.e., there is no other $\L_i$ projection of $y$ between $y_1$ and $y_2$
along $Q$.
See Figure~\figref{JunctionProof}, which illustrates these two
$\L_i$ projections.
We aim to establish that $\d_{P_i}(y,y_2) \le \d_{\L_i}(y,y_2)$,
which will imply equality
(because all projection segments from $y$ have the same
length on $P_i$, and the same length on $\L_i$---details below).
This will imply that
the regions on $P_i$ and on $\L_i$ determined by $\{y, y_1, y_2\}$ are            
isometric.
Thus $y$ would have two projection segments $y y_1$ and $y y_2$
on $P_i$.  But these two segments derive from $x_t \in \G_{\L_i}$
corresponding to a point $y \in C_{P_i}$, and there is a second point of $\G_{\L_i}$
that corresponds to the same $y$, on ``the other side'' of the cut locus.
Thus there must be a third projection from $y$ to $y_3$ on $P_i$,
which implies that $y$ is a junction of $C_{P_i}$.
This situation is illustrated in Figure~\figref{JunctionProof},
where two points along $\G_{\L_i}$ derive from ``different sides'' of $a$.

Now we prove the claim that $\d_{P_i}(y,y_2) \le \d_{\L_i}(y,y_2)$.
Consider the vertices $q_1,\ldots,q_k$ between
$y_1$ and $y_2$, ordered from $y_1$ to $y_2$.
(In Figure~\figref{JunctionProof} there is just one such vertex $q_1$;
see Figure~\figref{Toponogov} for a generic example.)
Then the right triangle $\triangle y y_1 q_1$ is flat on $\L_i$,
and non-negatively curved on $P_i$.
Hence by the Alexandrov-Toponogov Comparison Theorem
(Lemma~\lemref{comp_2}),
the hypotenuse is no longer on $P_i$,
$\d_{P_i}(y,q_1) \le \d_{\L_i}(y,q_1)$,
and the angle at $q_1$ is at least as large on $P_i$
(Lemma~\lemref{comp_1}),
$\angle_{P_i} (y, q_1, y_1) \ge \angle_{\L_i} (y, q_1, y_1)$.
This angle inequality implies
that $\angle_{P_i} (y, q_1, q_2) \le \angle_{\L_i} (y, q_1, q_2)$;
see Figure~\figref{Toponogov}.
\begin{figure}[htbp]
\centering
\includegraphics[width=0.5\linewidth]{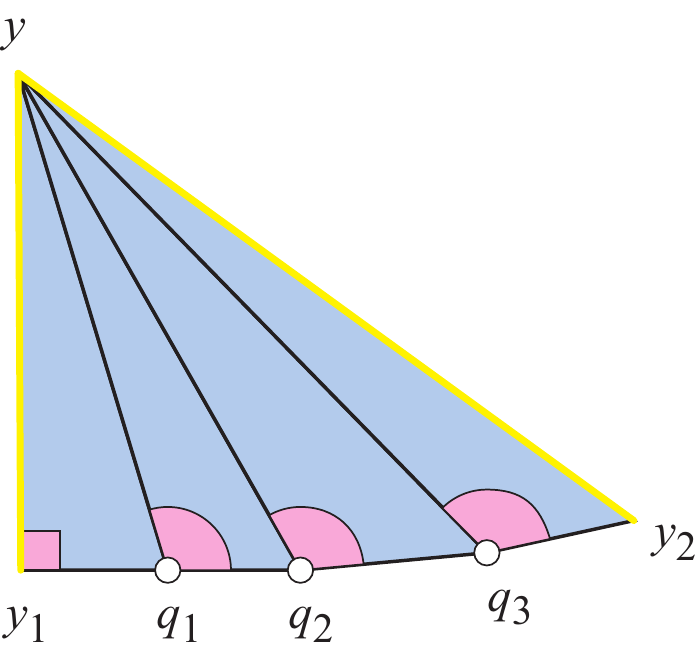}
\caption{Applying the Alexandrov-Toponogov Comparison Theorem 
to
the sequence of triangles apexed at $y$ along $Q$ from $y_1$ to $y_2$.}
\figlab{Toponogov}
\end{figure}
Continuing with the same logic,
$\d_{P_i}(y,q_2) \le \d_{\L_i}(y,q_2)$,
and
$\angle_{P_i} (y, q_2, q_1) \ge \angle_{\L_i} (y, q_2, q_1)$,
so
$\angle_{P_i} (y, q_2, q_3) \le \angle_{\L_i} (y, q_2, q_3)$.
This leads (by induction) to the conclusion
that the last distance is no longer on $P_i$:
$\d_{P_i}(y,y_2) \le \d_{\L_i}(y,y_2)$,
which is exactly what we aimed to establish.
We now show that there must be equality here.

Suppose instead some inequality in the chain of reasoning
above were strict.  This leads to 
$\d_{P_i}(y,y_2) < \d_{\L_i}(y,y_2)$.
But we already know that 
$\d_{\L_i}(y,y_2) = \d_{\L_i}(y,y_1) = \d_{P_i}(y,y_1)$,
and so $\d_{P_i}(y,y_2) < \d_{P_i}(y,y_1)$.
But this is a contradiction, because all projections from
$y$ to $Q$ on $P_i$ have the same minimal length.

Thus our conclusion above that there must be a third projection from $y$ to $y_3$ on $P_i$ follows, which means that
that $y$ is a junction of $C_{P_i}$.
{\hfill\ABox}\end{proof}


Having established the subpeel nesting property, we obtain next the same for peels.
Outside the isometric regions of $P_i$ and $\L_i$
that correspond to the external common subtrees, $\G_{\L_i}$ is in one-to-one
correspondence with what remains from $Q$ 
(because of the projection onto along segments).
I.e., each time $\G_{\L_i}$ enters a peel of $\L_i$, its part inside that peel is in one-to-one correspondence with the part of $Q$ bounding that peel minus 
the regions of $P_i$ and $\L_i$ that correspond to the external common subtree
(because of the projection along segments).

So visiting the same peel of $\L_i$ twice
(necessarily outside these isometric regions of $P_i$ and $\L_i$) 
would produce a contradiction with this bijective correspondence.

\noindent
We summarize the main import of this section in the following lemma.

\begin{lemma}[Cone Embedding]
Let the curve $Q$ on $\P$ live on a cone $\L_i$ on one side.
Then the corresponding half-surface $P_i$ can be isometrically 
embedded into that cone 
when it is cut along all edges of $C_{P_i}$ not incident to $Q$.
\lemlab{Embed}
\end{lemma}
\begin{proof}
Because $C_{P_i}$ spans the vertices of $P_i$, the cutting removes all
curvature and leaves a locally flat surface.
The peels of $P_i$ may then be embedded
within the cone $\L_i$ via Lemma~\lemref{Nesting},
which guarantees non-overlapping.
{\hfill\ABox}\end{proof}

\section{Source Unfolding}
\seclab{SourceUnfolding}
Henceforth let $\overline{P_i}$ be the embedded image of $P_i$ on the cone $\L_i$ 
given by Lemma~\lemref{Embed}. 
The proof of Theorem~\thmref{HalfSourceUnfolding} below
needs one more lemma, concerning cutting and flattening $\overline{P_i}$.


\begin{lemma}[Generator Cut]
Let $g$ be a point of $Q$ closest on the cone $\L_i$ to the apex $a$ of the cone.
Then cutting $\L_i$ along the generator $ag$ unfolds $\overline{P_i}$ to
one piece in the plane, i.e., the cut does not disconnect $\overline{P_i}$.
\lemlab{GenCut}
\end{lemma}
\begin{proof}
We will prove that $ag$ intersects only one peel $\a_1$ of $\overline{P_i}$.
In any case, let $\a_1$ be the last peel (from $a$) of $\overline{P_i}$ intersected
by $ag$.  
Let $y$ be the closest point to $a$ in $ag \cap \alpha_1$.
Then, in $\alpha_1$, $g$ is the point in $Q$ closest to $y$, so
$ag$ follows a shortest path segment within $\alpha_1$, projecting to $g$.
By the peel nesting property, this shortest path segment is included
in $\alpha_1$; therefore, cutting along $yg$ doesn't disconnect $\overline{P_i}$.

Suppose now that $ag$ meets another peel $\a_2$ of $\overline{P_i}$ at point $x$,
which projects in $\a_2$ (by a shortest path) to $x'$ on $Q$.
Because $g \in \a_1$ and $x' \in \a_2$ and $\a_1$, and because $\a_2$ are
distinct, $g$ and $x'$ are distinct.
Thus $|x x'| < |x g|$, where, 
to ease notation, we use $|p q|$ to represent $\d_{\L_i}(p,q)$.

This inequality contradicts the assumption that $ag$ is a shortest path
to $Q$:
$$
|ax'| \le |ax| + |xx'| < |ax| + |xg| = |ag| \; .
$$
The first inequality follows from the triangle inequality
on $\L_i$, which is itself a complete metric space.

Knowing that $ag$ meets just one peel of $\overline{P_i}$ shows that it
does not disconnect $\overline{P_i}$, and the lemma claim is established.
{\hfill\ABox}\end{proof}

We now have assembled all the machinery needed to establish our first main theorem:

\begin{theorem}[Half Source Unfolding]
For any $Q$ that lives on a cone $\L_i$ to the $P_i$-side,
such that each generator of $\L_i$ meets $Q$ in one point,
the source unfolding of the corresponding half $P_i$ of $\P$
is non-overlapping.
\thmlab{HalfSourceUnfolding}
\end{theorem}
\begin{proof}
First, cut all the edges of $C_{P_i}$ not incident to $Q$.
Then 
apply Lemma~\lemref{Embed} to obtain the embedding $\overline{P_i}$
in $\L_i$.
Finally, cut the generator $ag$ to a closest point $g \in Q$
and unfold $\overline{P_i}$ by Lemma~\lemref{GenCut} into the plane.
{\hfill\ABox}\end{proof}

Theorem~\thmref{HalfSourceUnfolding} can be viewed as a 
significant generalization
of the result in~\cite{o-eumap-08}, which established it
for ``medial axis polyhedra,'' 
whose base edges form $Q$ and whose lateral edges constitute the 
``upper component'' of the cut locus.

The unfolding of the $P_1$ half of the truncated cube
in Figure~\figref{CubeTrunc}(d) best illustrates this
theorem.


We now turn to joining the two halves.
Several strategies are available, and we select a simple one.
We refer again to Figure~\figref{ConvexQUnf}.

\begin{theorem}[Full Source Unfolding]
For $Q$ a convex curve in the class of curves we described,
additional cuts permit joining the halves to one, simple polygon.
\thmlab{SourceUnfolding}
\end{theorem}
\begin{proof}
Let $P_1$ be the half of $\P$ to the
convex side of $Q$, and $P_2$ the half to the (possibly) nonconvex side.
We unfold $P_1$ just as described in Theorem~\thmref{HalfSourceUnfolding} above.
Note that the planar image of $Q$ is a convex curve, as all
of $Q$'s vertices are convex to the $P_1$-side.
To the nonconvex side $P_2$, we cut all of $C_{P_2}$, including those
edges incident to $Q$.
Recall these edges will be incident to vertices that are convex vertices
to the $P_2$-side.
In addition, we cut shortest path 
segments from $C_{P_2}$ to the other (reflex) vertices of $Q$
(none if $Q$ is a quasigeodesic, but possibly several if $Q$ is merely
convex to one side).
These combined cuts partition $P_2$ into polygonal regions $R_i$ 
($i=0,1,2,\ldots$), each of which projects onto its base $q_{i-1} q_i$ on $Q$.
See Figure~\figref{ConvexQUnf}.
The regions are separated by empty cones, whose bounding rays 
are separated by an angle equal to the curvature at $q_i$.
Because these curvatures are positive, and $Q$ is convex, the joined
pieces do not overlap.

It only remains to argue that the cut to $g$ on the $P_1$-side
does not coincide with a cut to a vertex $q_i$ on the $P_2$-side,
for that would disconnect the unfolding into two pieces.
This follows from~\cite[Cor.~1]{iiv-qfpcs-07},
which shows that $g$ could only be a corner of $Q$ if it were
reflex to the $P_1$-side, because $ag$ makes an angle at least
$\pi/2$ with $Q$ to each side of $ag$.
But we know that $Q$ is convex to the $P_1$-side, so it cannot be
that $g=q_i$.
{\hfill\ABox}\end{proof}



For a convex curve $Q$ shrinking to a point $x$,
the full source unfolding with respect to $Q$ approaches the point source unfolding
with respect to $x$,
as one would expect.

\section{Future Work}
\seclab{Future}
We have not yet addressed the computational complexity of
constructing the source unfolding from a given $Q$,
but we expect it will be polynomial in the complexity of $Q$.

A secondary issue is finding a $Q$ that satisfies our conditions.
Although it is known that every convex polyhedron has at least three
distinct
simple closed quasigeodesics,
there is no polynomial-time algorithm known for finding
one~\cite[Prob.~24.2]{do-gfalop-07}.
However, it is easy to find convex curves through at most one vertex:
for example, a convex curve inside any face, or the curve obtained
by truncating any vertex $v$ of $\P$ orthogonal to a vector within
the tangent cone of $v$ (e.g., Figure~\figref{TetraSlice}(a)).
Also, one can construct
quasigeodesic
loops through at most one vertex by a minor modification of
the technique described in~\cite{iov-sucpql-10}.
Such a curve satisfies conditions~(a) and~(b) in
Section~\secref{Introduction}.

Finally, we leave unresolved
determining the largest class of curves $Q$ for which 
Theorems~\thmref{HalfSourceUnfolding} and~\thmref{SourceUnfolding} hold.

{\bf Acknowledgments.}
We are grateful to the remarkably perceptive comments (and the patience)
of Stefan Langerman and several other anonymous referees.

\bibliographystyle{alpha}
\bibliography{/Users/orourke/bib/geom/geom}

\newcommand{\etalchar}[1]{$^{#1}$}
\begin{thebibliography}{ACC{\etalchar{+}}08}

\bibitem[ACC{\etalchar{+}}08]{accddlopt-calgs-08}
Zachary Abel, David Charlton, Sebastien Collette, Erik~D. Demaine, Martin~L.
  Demaine, Stefan Langerman, Joseph O'Rourke, Val Pinciu, and Godfried
  Toussaint.
\newblock Cauchy's arm lemma on a growing sphere.
\newblock Technical Report~90, Smith College, April 2008.
\newblock arXiv.0804.0986v1 [cs.CG].

\bibitem[Ale06]{a-igcs-06}
Aleksandr~D. Alexandrov.
\newblock {\em Intrinsic Geometry of Convex Surfaces}.
\newblock Chapman \& Hall/CRC, Boca Raton, FL, 2006.
\newblock A. D. Alexandrov Selected Works. Edited by S.~S.~Kutateladze.
  Translated from the Russian by S.~Vakhrameyev.

\bibitem[AZ67]{az-igs-67}
Aleksandr~D. Alexandrov and Victor~A. Zalgaller.
\newblock {\em Intrinsic Geometry of Surfaces}.
\newblock American Mathematical Society, Providence, RI, 1967.

\bibitem[DO07]{do-gfalop-07}
Erik~D. Demaine and Joseph O'Rourke.
\newblock {\em Geometric Folding Algorithms: Linkages, Origami, Polyhedra}.
\newblock Cambridge University Press, July 2007.
\newblock \url{http://www.gfalop.org}.

\bibitem[IIV07]{iiv-qfpcs-07}
Kouki Ieiri, {Jin-ichi} Itoh, and Costin V\^{i}lcu.
\newblock Quasigeodesics and farthest points on convex surfaces.
\newblock Submitted, 2007.

\bibitem[IOV09]{iov-sucpr-09}
{Jin-ichi} Itoh, Joseph O'Rourke, and Costin V\^{i}lcu.
\newblock Source unfoldings of convex polyhedra with respect to certain closed
  polygonal curves.
\newblock In {\em Proc. 25th European Workshop Comput. Geom.}, pages 61--64.
  EuroCG, March 2009.
\newblock Full version submitted to a journal, May 2009.

\bibitem[IOV10]{iov-sucpql-10}
{Jin-ichi} Itoh, Joseph O'Rourke, and Costin V\^{i}lcu.
\newblock Star unfolding convex polyhedra via quasigeodesic loops.
\newblock {\em Discrete Comput. Geom.}, 44:35--54, 2010.

\bibitem[O'R08]{o-eumap-08}
Joseph O'Rourke.
\newblock Edge-unfolding medial axis polyhedra.
\newblock In {\em Proc. 24th European Workshop Comput. Geom.}, pages 103--106,
  March 2008.

\bibitem[OV11]{ov-ceccc-11}
Joseph O'Rourke and Costin V\^{i}lcu.
\newblock Conical existence of closed curves on convex polyhedra.
\newblock \url{http://arxiv.org/abs/}, February 2011.

\bibitem[Piz07]{p-catge-07}
Paolo Pizzetti.
\newblock Confronto fra gli angoli di due triangoli geodetici di eguali lati.
\newblock {\em Rend. Circ. Mat. Palermo}, 23:255--264, 1907.

\bibitem[SS86]{ss-spps-86}
Micha Sharir and Amir Schorr.
\newblock On shortest paths in polyhedral spaces.
\newblock {\em SIAM J. Comput.}, 15:193--215, 1986.

\bibitem[Thu98]{t-spts-98}
W.P. Thurston.
\newblock {Shapes of polyhedra and triangulations of the sphere}.
\newblock {\em Geometry and Topology Monographs}, 1:511--549, 1998.

\end{thebibliography}
\normalsize

\newpage
\section*{Appendix: \\Alexandrov-Toponogov Comparison Theorem}
Several well-known comparison results for convex surfaces
are usually identified as variants of
Toponogov's 1959 theorem.
For
triangles on convex surfaces, however, 
Alexandrov 
proved them, without any assumption of differentiability,
as early as 1948,
and Pizzetti~\cite{p-catge-07} considered the differentiable case in 1907.
See~\cite[p.242]{a-igcs-06} or
\cite[p.32]{az-igs-67} for versions in English; 
see also
\cite{accddlopt-calgs-08}.

Essentially, the results compare triangles or ``hinges''
on a given surface to those in the plane.

A {\it triangle} in a convex surface is a collection of three segments 
$\gamma_1,\gamma_2,\gamma_3$ such that $\gamma_i$ and $\gamma_{i+1}$ have the common endpoint $a_{i+2}$ (indices mod 3).
We shall denote the triangle by $\gamma _1 \gamma _2 \gamma _3$ or, 
if the segments are clear from the context, by $a_1 a_2 a_3$. 
We use $\lambda \left(\gamma\right)$ to denote the length of the curve $\gamma$.

The first lemma says that if you draw a triangle in the plane with the
same lengths as a triangle on a convex surface, the planar triangle
angles in general get smaller:
they are at most as large as the convex-surface angles.

\begin{lemma}
\lemlab{comp_1}
For any triangle $\gamma_1 \gamma_2 \gamma_3$ in a convex surface there exists a planar triangle ${\overline\gamma_1} {\overline\gamma_2} {\overline\gamma_3}$ with 
$\lambda \left(\gamma_i\right) = \lambda \left({\overline\gamma_i}\right)$. We have
$\angle{\overline\gamma_i} {\overline\gamma_{i+1}} \le 
\angle \gamma_i \gamma_{i+1}$, $i=1,2,3$ ({\rm mod} $3$), and equality holds if and
only if $\gamma_1 \gamma_2 \gamma_3$ is isometric to 
${\overline\gamma_1} {\overline\gamma_2} {\overline\gamma_3}$.
\end{lemma}

A {\it hinge} is a pair of segments,
$\gamma_1$ from $a$ to $b$ and 
$\gamma_2$ from $a$ to $c$, and 
the angle $\angle bac$ included between them at $a$.
We denote the hinge by $\gamma _1 \gamma _2$. 

The second lemma says that if you draw a hinge in the plane with the same
angle as a hinge on a convex surface, the planar hinge endpoint separation
in general gets larger.
Next, we denote by $\overline \d$ the Euclidean distance, and by $\d$ the distance on the surface.

\begin{lemma}
\lemlab{comp_2}
For any hinge $\gamma_1 \gamma_2$ in a convex surface $S$ there exists a planar  hinge ${\overline\gamma_1} {\overline\gamma_2}$ with 
$\angle {\overline b} {\overline a} {\overline c} = \angle bac$. We have
$\d (b,c) \le \overline \d ({\overline b}, {\overline c})$, and equality holds if and only if there exists a segment joining $b$ to $c$ in $S$ such that $abc$ is isometric to ${\overline a} {\overline b} {\overline c}$.
\end{lemma}

\end{document}